\documentclass[10pt]{amsart}

\usepackage{amsfonts}
\usepackage{graphicx}
\usepackage[english]{babel}
\usepackage[left=3cm,right=3cm,top=2.5cm,bottom=3cm]{geometry}
\usepackage{amsmath,amssymb,amsthm,bbm,color,graphics,version}
\usepackage{mathrsfs}
\usepackage{cite}
\usepackage{latexsym,amsfonts,amscd,indentfirst}
\usepackage{float}
\usepackage[numbers,sort&compress]{natbib}
\usepackage{epsfig,graphics,picinpar,subfigure}

\newtheorem{theorem}{Theorem}
\newtheorem{assumption}{Assumption}

\newtheorem{condition}[theorem]{Condition}
\newtheorem{conjecture}[theorem]{Conjecture}

\newtheorem{definition}[theorem]{Definition}
\newtheorem{example}[theorem]{Example}

\newtheorem{lemma}[theorem]{Lemma}

\newtheorem{prop}[theorem]{Proposition}
\newtheorem{remark}[theorem]{Remark}

\newcommand{\bearno}{\begin{eqnarray*}}
\newcommand{\enarno}{\end{eqnarray*}}

\title{A note on optimal expected utility of dividend payments with proportional reinsurance}

\author{Xiaoqing Liang}
\address{School of Sciences, Hebei University of Technology, Tianjin 300401, P.R. China}
\email{liangxiaoqing115@hotmail.com}

\author{Zbigniew Palmowski}
\address{Faculty of Pure and Applied Mathematics,
Wroc\l aw University of Science and Technology,
Wyb. Wyspia\'nskiego 27, 50-370 Wroc\l aw, Poland}
\email{zbigniew.palmowski@gmail.com}

\thanks{This work is partially supported by the National Science Centre under the grant 2015/17/B/ST1/01102.The first author is partially supported by the National Natural Science Foundation of China (11571189) and High School National Science Foundation of Hebei Province (QN2016176). Both of the authors kindly acknowledge partial support by the project RARE -318984, a Marie Curie IRSES Fellowship within the 7th European Community
Framework Programme.}

\date{\today}

\allowdisplaybreaks
\begin{document}

\begin{abstract}
In this paper, we consider the problem of maximizing the expected  discounted utility of dividend payments
for an insurance company that controls risk exposure by purchasing proportional reinsurance. We assume the preference of the insurer is of CRRA form. By solving the corresponding Hamilton-Jacobi-Bellman equation, we identify the value function and the corresponding optimal strategy. We also analyze the asymptotic behavior of the value function for large initial reserves. Finally, we provide some numerical examples to illustrate the results and analyze the sensitivity of the parameters.

\vspace{3mm}


\end{abstract}

\maketitle

\pagestyle{myheadings} \markboth{\sc X.\ Liang --- Z.\ Palmowski}
{\sc Optimal expected utility of dividend payments with proportional reinsurance}

\vspace{1.8cm}

\tableofcontents

\newpage

 {\section{Introduction}}
In recent years there has been increasing attention towards the
utilization of stochastic control theory to insurance-related problems.
This is due to the fact that a company, such as a property-liability
insurance company or a pension-fund management company, can control
reinsurance strategies or investment strategies and can pay dividends
to maximize (or minimize) a certain objective function under different
constraints. Two kinds of risk processes have been considered.
The first one concerns classical Cram\'er-Lundberg process being drift process minus
compound Poisson process, see e.g. \citet{Buhlmann}, \citet{Hipp}, \citet{azcue2005optimal, azcue2015optimal}.
Later this case was generalized to the spectrally negative L\'evy risk process, see
\citet{APP, APP2}, \citet{KPal}, \citet{Loeffen1, loeffen2009optimal}, \citet{LR} and references therein.
The second risk process, considered also in this paper,
is a diffusion surplus risk model. In this model, the liquid asset processes
of the corporation are driven by Brownian motion with constant drift
and diffusion coefficients. The drift term corresponds to the expected
profit per unit time, while the diffusion term is interpreted as risk.
The classic studies on this subject are those by
\citet{GerberEllias}, \citet{JeanShir}, \citet{Cade}, \citet{asmussen1997controlled}, \citet{asmussen2000optimal}, \citet{bai2010optimal},
\citet{hojgaard1999controlling, hojgaard2004optimal}, \citet{paulsen2003optimal, paulsen2008optimal}, \citet{Zhou}, \citet{david2005minimizing} and many others.
The details can be found in the survey paper \citet{Hansjoerg} and in the book \citet{Schmbook}.

The goal of this paper is to maximize the expected  discounted utility of dividend payments
for an insurance company whose reserve evolves in time according to a diffusion process
and which controls risk exposure by purchasing proportional reinsurance.
That is, in this paper we formally consider the following optimization problem.

Let ($\Omega,\ \mathcal F,\ P$) denote a complete probability space endowed with information filtration ${\{{\mathcal {F}}_{t}\}}_{t\geq 0}$ and $\{B_t\}_{t\geq 0}$ be a standard Brownian motion adapted to the filtration.
Let $R$ be a risk process being a Brownian motion with drift:
\begin{eqnarray}\label{R}
R_{t}= (1+\theta) a t-Y_t, ~~~~~~~~R_0=x
\end{eqnarray}
for the aggregate cumulative amount of claims counted up to time $t$:
\begin{eqnarray*}
d Y_{t}=a dt-b d B_{t},~~~~~~~~~Y_0=0,
\end{eqnarray*}
where $a$ and $b$ are positive constants, $x\geq 0$ is the initial surplus, $(1+\theta)a$
is the premium rate with the safety loading $\theta >0$.

Apart of the risk process $R_t$ we will consider the dividend payments.
Let $C=(C_t)_{t\geq 0}$ be an adapted and nondecreasing process representing all accumulated dividend payments up to time $t$.
In our model we assume that $(C_t)_{t\geq 0}$ is absolutely continuous with respect to Lebesgue measure.
Hence we suppose that the process $C$ admits almost surely a density process denoted by $c_t \geq 0$
modeling the intensity of the dividend payments in continuous time.

In our model, we add another new and very important feature in the context of dividend payments with utility function.
We consider reinsurance policy, in which
part of the premium rate $(1+\eta)q_ta$ for some proportion $q_t\in [0,1]$
is diverted to some reinsurer who will cover $q_t$ of arrived claims $Y_t$.
In this way, the insurance company can reduce its risk exposure
and therefore the reinsurance has been extensively studied, see for example
\citet{asmussen2000optimal}, \citet{azcue2005optimal}, \citet{choulli2003diffusion}, \citet{chen2013optimal},
\citet{hojgaard1999controlling, hojgaard2004optimal}, \citet{liang2012dividends}, \citet{zhou2012optimal},
and references therein.

Thus the reserve process $X^{\pi}_t$ evolves as follows:
\begin{eqnarray}
d X^{\pi}_t&=&((1+\theta)a-(1-q_t)a ) d t+b(1-q_t)d B_t-(1+\eta)q_ta dt-c_t dt \nonumber\\
&=&(\theta-\eta q_t)a d t+b(1-q_t)d B_t-c_t dt,\label{control}
\end{eqnarray}
where $\pi$ in the superscript denotes a strategy which is described by a two-dimensional stochastic process $(q_t, c_{t})$
that supposed to be chosen in optimal way, where the criterium for the optimality will be specified later.
We will assume that $\eta\geq \theta$. When $\eta>\theta$, the fraction of the premiums diverted to the reinsurer is larger than that of each claim covered by the reinsurer,
this is called non-cheap reinsurance.
When $\eta=\theta$, we say the reinsurance is cheap.
Both of these cases will be considered in this paper.

We observe the regulated process $X_t$ until the time of ruin:
\begin{equation*}
\tau=\inf\{t\geq 0\colon X^\pi_t<0\}.
\end{equation*}

We define the target value function as
 \begin{equation}
 V(x)=\max_{\pi\in \Pi}\mathbb{E}_x\left[\int^{\tau}_0 e^{-\beta s}u(c_s)\;d s\right],
 \label{ORID-value-function}
 \end{equation}
where $\beta (> 0)$ is a discount factor, $u$ is some fixed utility function, $\mathbb{E}_x$ means expectation with respect to $\mathbb{P}_x(\cdot)=P(\cdot|X_0=x)$
and maximum is taken over all admissible strategies $\Pi$.
A strategy $\pi$ is said to be admissible if
 $(q_t, c_{t})$ is $\mathcal F$$_{t}$-progressively measurable and satisfies $0 \leq q_t\leq
  1, c_t\geq 0$ for all $t\geq0$. Finally, we assume that the ruin cannot be caused by the dividend payment.
Usually it is assumed that $u:\mathbb{R}_{\geq 0}\to\mathbb{R}_{\geq 0}$ is differentiable and nonnegative.
%, strictly increasing, strictly concave and $u(0)=0$.

In fact, in this paper we will consider only the Constant Relative Risk Aversion (CRRA) utility function:
\begin{equation}\label{assumption}
u(c)=\frac{c^p}{p},\qquad p\in(0,1).
\end{equation}

For above dividend problem we will prove the verification theorem
producing the Hamilton-Jacobi-Bellman (HJB) equation for optimal value function.
This is done in Section \ref{sec:HJB}.
Later (under some additional technical assumptions) we will solve it producing optimal strategy
which appear to be a Markovian one, that is $q_t=q(X^\pi_t)$ and $c_t=c(X^\pi_t)$
for some functions $q$ and $c$ given explicit.
In particular, for utility function \eqref{assumption}
we will prove that when the reserve is sufficiently small (less than
identified level $x^*$), insurance company is willing to buy part of reinsurance as well as diverting part of premium.
When the reserve is larger than $x^*$ then
the insurance company is able to afford all arrived claims and
the optimal strategy excludes the purchase of reinsurance.
Following \citet{hubalek2004optimizing} we also identify the optimal strategy for large reserves as $x\rightarrow \infty$.
Both separate cases of non-cheap and  cheap reinsurance polices are considered
in Sections \ref{sec:noncheap} and \ref{sec:cheap}. We also present some numerical examples to illustrate the results in Section \ref{sec:numericalexamples}.
We ends our paper by conclusions \ref{sec:conclusions}.

{\section{HJB equation}\label{sec:HJB}}
 For nonnegative $v\in \mathcal{C}^2$ with our optimization problem we associate the following
 Hamilton-Jacobi-Bellman (HJB) equation:
 \begin{eqnarray}
 & &\max_{0\leq q\leq1,c\geq 0}\bigg\{(\theta-\eta q)a v{'}(x)
 +\frac{1}{2}b^2(1-q)^2v{''}(x)-c v{'}(x)+u(c)-\beta v(x)
 \bigg\}=0,
 \label{ORID-HJB}
 \end{eqnarray}
 with the boundary condition
\begin{eqnarray}
v(0)=0;          \label{ORID-HJB-boundary}
 \end{eqnarray}
see \citet{fleming2006controlled} for a beautiful overview.
From now on we look only for those solutions $v$ of above HJB equation that
can be equal to the value function $V$. In other words, we exclude those solutions that for any reasons
cannot be a value function.

We start from the classical Verification Lemma.
\begin{lemma}\label{ORID-verif-lemma}
Suppose that $v(x)\in C^{2}$ is a nonnegative solution of the HJB equation (\ref{ORID-HJB}) with the boundary condition (\ref{ORID-HJB-boundary}). Then $v(x)\geq V(x)$ on $(0, \infty)$
where $V$ is the value function defined in (\ref{ORID-value-function}).
\end{lemma}
\begin{proof}
Fix $x>0$ and choose a pair of admissible strategies $c(\cdot)$ and $q(\cdot)$.
Select $0<\xi_1<x<\xi_2$.
\\Define
\[\tau_x=\inf\{t\geq 0: X^{\pi}_t=x\}\]
and
\[w(t,x)=e^{-\beta t}v(x).\]
By applying It${\rm \hat{o}}$'s formula to the process $w(t,X^{\pi}_t)$, we obtain
\begin{align}
&w\left(t\wedge \tau_{\xi_1}\wedge \tau_{\xi_2},X^{\pi}_{t\wedge \tau_{\xi_1}\wedge \tau_{\xi_2}}\right)
\nonumber\\
&=v(x)+\int^{t\wedge \tau_{\xi_1}\wedge \tau_{\xi_2}}_0 e^{-\beta s}
\left\{(\theta-\eta q_s)a v{'}(X^{\pi}_s)
 +\frac{1}{2}b^2(1-q_s)^2v{''}(X^{\pi}_s)-c_s v{'}(X^{\pi}_s)-\beta v(X^{\pi}_s)\right\}ds \nonumber\\
&~~~~+\int^{t\wedge \tau_{\xi_1}\wedge \tau_{\xi_2}}_0 e^{-\beta s}v{'}(X^{\pi}_s)b(1-q_s)dB_s\nonumber\\
&\leq v(x)-\int^{t\wedge \tau_{\xi_1}\wedge \tau_{\xi_2}}_0 e^{-\beta s}
u(c_s)ds+\int^{t\wedge \tau_{\xi_1}\wedge \tau_{\xi_2}}_0 e^{-\beta s}v{'}(X^{\pi}_s)b(1-q_s)dB_s, \label{itow}
\end{align}
where the inequality is derived from the HJB equation \eqref{ORID-HJB}.
Since the reinsurance strategy $q_s$ lies between $0$ and $1$, therefore the stochastic
integration term of \eqref{itow} is actually a martingale. Taking expectation on both sides of \eqref{itow}, we get
\begin{align}
\mathbb{E}_x\left[w\left(t\wedge \tau_{\xi_1}\wedge \tau_{\xi_2},X^{\pi}_{t\wedge \tau_{\xi_1}\wedge \tau_{\xi_2}}\right)+\int^{t\wedge \tau_{\xi_1}\wedge \tau_{\xi_2}}_0 e^{-\beta s}
u(c_s)ds\right]\leq v(x). \label{eq:conv}
\end{align}
Now let $\xi_1 \downarrow 0$, $\xi_2\rightarrow \infty$ and $t\rightarrow \infty$, then $t\wedge \tau_{\xi_1}\wedge \tau_{\xi_2}\uparrow \tau$ almost surely. Thus by applying Fatou's lemma for the first term of the expectation in \eqref{eq:conv}, we have
\[\liminf\mathbb{E}_xw\left(t\wedge \tau_{\xi_1}\wedge \tau_{\xi_2},X^{\pi}_{t\wedge \tau_{\xi_1}\wedge \tau_{\xi_2}}\right)\geq \mathbb{E}_x\left[\liminf w\left(t\wedge \tau_{\xi_1}\wedge \tau_{\xi_2},X^{\pi}_{t\wedge \tau_{\xi_1}\wedge \tau_{\xi_2}}\right)\right]=0,\]
and by applying the monotone convergence theorem, the second term of the expectation in \eqref{eq:conv} converges to
\[\mathbb{E}_x\left[\int^{\tau}_0 e^{-\beta s}
u(c_s)ds\right].\]
Hence, we can conclude that:
\begin{align*}
\mathbb{E}_x\left[\int^{\tau}_0 e^{-\beta s}
u(c_s)ds\right]\leq v(x).
\end{align*}
Maximization over all admissible strategies $c(\cdot)$ and $q(\cdot)$ gives the final result.
\end{proof}

Note that if \eqref{assumption} holds, then supremum in \eqref{ORID-HJB} without constraints
is realized for
\begin{eqnarray}
q(x)&=&1+\frac{\eta a}{b^2}\frac{v{'}(x)}{v{''}(x)}\label{ORID-rein-prop-rate}
\end{eqnarray}
and
\begin{equation}\label{func}
c(x)=(v{'}(x))^{-\frac{1}{1-p}}.
\end{equation}

Taking strategy $\pi=(q(X^\pi_t), c(X^\pi_t))$
makes process $X^\pi_t$ to be a diffusion.

Denote:
\begin{align}
\alpha=1+\frac{2b^2\beta}{\eta^2a^2}. \label{ORID-alpha}
\end{align}

We will solve our optimization problem under two family of assumptions:
\begin{assumption}\label{thm1a}
$\eta>\theta$, $\alpha(1-p)>1$ and $(2-\alpha(1-p))\eta-2\theta<0$
\end{assumption}
which is for non-cheap reinsurance and
\begin{assumption}\label{thm2a}
$\eta=\theta$ and $\alpha(1-p)>1$
\end{assumption}
which is for cheap reinsurance.

For the optimal value function $V$, we can also derive the following property, a similar result can be found in \citet{choulli2003diffusion}.
\begin{lemma}
The optimal value function $V$ is increasing and strictly concave.
\end{lemma}
\begin{proof}
It is obvious to see that $V$ is increasing. To prove the concavity we will follow the idea of the proof of \cite[Prop.2]{choulli2003diffusion}. 
From the dynamic programming
principle, we know $V(x)$ satisfies the following equality
\begin{align}
V(x)=\max_{\pi\in\Pi}\mathbb{E}_x\left[\int^{\tau_y}_{0} e^{-\beta s}u(c_s)d s+e^{-\beta \tau_y}V(y)\right],
\label{dpp_v}
\end{align}
where $0<y<x$ and
$$\tau_y=\inf\left\{t:X^{\pi}_t=y\right\}.$$
For $h>0$, let $\Pi^{h}$ be the set of strategies $\pi$
such that
\begin{align*}
\int^{\zeta}_0(\theta-\eta q_s)a d s+\int^{\zeta}_0 b(1-q_s)d B_s-\int^{\zeta}_0 c_s ds=-h
\end{align*}
on the set $\{\zeta<\infty\}$, where $\zeta$ is a stopping time defined by
\begin{align*}
\zeta=\inf\left\{t\geq0: \int^{t}_0(\theta-\eta q_s)a d s+\int^{t}_0 b(1-q_s)d B_s-\int^{t}_0 c_s ds=-h\right\}.
\end{align*}
By putting $h=x-y$ (hence $y=x-h$) from (\ref{dpp_v}) we obtain that
\begin{align}
V(x)=\max_{\pi\in\Pi^h}\mathbb{E}_x\left[\int^{\zeta}_{0} e^{-\beta s}u(c_s)d s+e^{-\beta \zeta}V(y)\right].
\label{dpp_v2}
\end{align}
Then
\begin{align}
V(x)-V(x-h)=\max_{\pi\in\Pi^h}\mathbb{E}_x\left[\int^{\zeta}_{0} e^{-\beta s}u(c_s)d s+(e^{-\beta \zeta}-1)V(x-h)\right].
\label{diff_v}
\end{align}
Since $e^{-\beta \zeta}<1$ and $V(x)$ is nondecreasing function of $x$, the right-hand side of (\ref{diff_v}) is a decreasing function of $x$. Thus $V(x)-V(x-h)$ is decreasing in $x$. Hence, $V{'}(x)$ is also decreasing and $V$ is strictly concave.
\end{proof}

As the analysis above, the idea to solve the original optimization problem is to first find an increasing, concave and nonnegative solution $v$ of HJB equation (\ref{ORID-HJB}) with the boundary condition (\ref{ORID-HJB-boundary}), and then construct strategies such that the derived function $v$ can be realized by them. Thus the reverse inequality of Lemma \ref{ORID-verif-lemma} holds for this candidate solution $v$. Hence, $v$ is indeed the value function we are looking for.

{\section{ Non-cheap reinsurance }\label{sec:noncheap}}
\indent
We will show that there exists a point $x^*$ such that
\begin{equation}\label{xstar}
q(x)=0\quad \mbox{for $x>x^*$.}\end{equation}
It implies that when the wealth of the insurance company is larger than $x^*$ then
the optimal strategy will be to not reinsure the arrive claims. In other words,
the insurance company can afford to cover all the claims by itself.
Therefore, we analyze the value function in two intervals $(0,x^*]$ and $(x^*, \infty)$, respectively.
Let
\begin{align}
B&=-\frac{2b^2}{\eta^2a^2}\frac{1-p}{1-\alpha+\alpha p}>0, \label{ORID-B}\\
D&=\frac{2b^2}{\alpha \eta^2 a}(\eta-\theta)>0. \label{ORID-D}
\end{align}
\begin{prop}\label{nocheap1}
Assume Assumption \ref{thm1a}. Then on $(0,x^*]$ the function $v$ solving \eqref{ORID-HJB} is given by
\begin{align}
v(x)=\frac{e^{-\xi}}{\beta}\left[-(\eta-\theta)a+\frac{\eta^2a^2}{2b^2}Be^{\frac{\xi}{1-p}}+
\frac{\eta^2a^2}{2b^2}D+\frac{1-p}{p}e^{\frac{\xi}{1-p}}\right],\label{ORID-value-function-noncheapa}
\end{align}
where $B$ and $D$ are given in (\ref{ORID-B}) and (\ref{ORID-D}) and
$\xi=g^{-1}(x)$, where $g^{-1}$ is the inverse of function $g$:
\begin{align}
g(\xi)=(1-p)Be^{\frac{\xi}{1-p}}+D\xi+Q_1, \label{ORID-fun:g}
\end{align}
for $Q_1=-(1-p)Be^{\frac{\xi_0}{1-p}}-D\xi_0$
and \begin{align}
\xi_0=(1-p)\ln \frac{(\eta-\theta)(\alpha-1-\alpha p)ap}{\alpha(1-p)^2}. \label{ORID-xi-0}
\end{align}
Above,
\begin{align}
x^*=g(\xi^*), \label{ORID-x^*}
\end{align}
where
\begin{align}\label{xistar}
\xi^*=(1-p)\ln\left(\frac{b^2-\eta a D}{\eta a B}\right).
\end{align}

%The corresponding dividend strategy and reinsurance proportion
%are given by
%\begin{eqnarray}
%c(x)=(v'(x))^{-\frac{1}{1-p}}, \nonumber
%\end{eqnarray}
%and
%\begin{equation}
%q(x)=1-\frac{\eta a}{b^2}g'(\xi), \nonumber
%\end{equation}
%respectively.
\end{prop}
\begin{proof}
Substituting \eqref{ORID-rein-prop-rate} and \eqref{func} into (\ref{ORID-HJB}) gives:
\begin{equation}
-(\eta-\theta)av'(x)-\frac{\eta^2a^2}{2b^2}\frac{(v'(x))^2}{v''(x)}+\frac{1-p}{p}(v'(x))^{-\frac{p}{1-p}}-\beta v(x)=0.
\label{ORID:HJB2}
\end{equation}
\par
We choose the following variable transform
\begin{align}
x=g(\xi),~~~~~v'(g(\xi))=e^{-\xi}. \label{transform}
\end{align}
Indeed, note that
\begin{equation}\label{trzy}
v''(g(\xi))=-\frac{e^{-\xi}}{g'(\xi)}.
\end{equation}
Plugging it into (\ref{ORID:HJB2}) produces
\begin{eqnarray}
-(\eta-\theta)ae^{-\xi}+\frac{\eta^2a^2}{2b^2}e^{-\xi}g'(\xi)+\frac{1-p}{p}e^{\frac{p}{1-p}\xi}-\beta v(g(\xi))
=0. \label{ORID-eq:g-xi}
\end{eqnarray}
Taking derivatives with $\xi$ and collecting terms, we derive
\begin{eqnarray}
\frac{\eta^2a^2}{2b^2}g''(\xi)-\left(\frac{\eta^2a^2}{2b^2}+\beta\right)g'(\xi)+(\eta-\theta)a+e^{\frac{\xi}{1-p}}=0.
\label{eq:g-xi-1}
\end{eqnarray}
Denote $h(\xi)=g'(\xi)$. Then the above equation can be rewritten as follows:
\begin{eqnarray}
h'(\xi)-\alpha h(\xi)+\frac{2b^2}{\eta^2a}(\eta-\theta)+\frac{2b^2}{\eta^2a^2}e^{\frac{\xi}{1-p}}=0
\label{ORID-eq}
\end{eqnarray}
for some $Q=h(0)$.
Solving above ordinary differential equation (ODE) we can observe that
\begin{eqnarray}
h(\xi)=Qe^{\alpha\xi}-\int^{\xi}_0 \frac{2b^2}{\eta^2 a}(\eta-\theta)e^{\alpha(\xi-u)}du
-\int^\xi_0 \frac{2b^2}{\eta^2a^2}e^{\alpha \xi}e^{(\frac{1}{1-p}-\alpha)u}du.
\end{eqnarray}
\par
Careful calculation and rearranging all terms give
\begin{eqnarray}
h(\xi)=A e^{\alpha\xi}+Be^{\frac{\xi}{1-p}}+D \nonumber
\end{eqnarray}
for
\begin{align}
A&=Q-\frac{2b^2}{\alpha \eta^2 a}(\eta-\theta)+\frac{2b^2}{\eta^2 a^2}\frac{1-p}{1-\alpha+\alpha p}. \nonumber
\end{align}
In view of \eqref{trzy} and the concavity of $v$, the value of $A$ should be nonnegative. Then
\begin{align}
h(\xi)=Ae^{\alpha \xi}+Be^{\frac{\xi}{1-p}}+D>0. \label{h}
\end{align}
Recall that $h(\xi)=g'(\xi)$  and \eqref{ORID-eq:g-xi}, we get
\begin{align}
v(x)&=\frac{e^{-\xi}}{\beta}\left[-(\eta-\theta)a+\frac{\eta^2a^2}{2b^2}g'(\xi)+\frac{1-p}{p}e^{\frac{\xi}{1-p}}\right]
\nonumber\\
&=\frac{e^{-\xi}}{\beta}\left[-(\eta-\theta)a+\frac{\eta^2a^2}{2b^2}Ae^{\alpha\xi}
+\frac{\eta^2a^2}{2b^2}Be^{\frac{\xi}{1-p}}+\frac{\eta^2a^2}{2b^2}D+\frac{1-p}{p}e^{\frac{\xi}{1-p}}\right].
\label{ORID-value-function-noncheap}
\end{align}
Note that if we allow a general values of $q_t$ without constrains then
we only increase the value function giving us the possibility of considering the limit as $x\to\infty$.
%Moreover, by Lemma \ref{lem:x-inf}, it is seen that the asymptotic solution for \eqref{ORID:HJB2} as
%$x\rightarrow \infty$
%is the same as that of the following equation
%\begin{align}
%-\frac{\eta^2a^2}{2b^2}\frac{(v'(x))^2}{v''(x)}+\frac{1-p}{p}(v'(x))^{-\frac{p}{1-p}}-\beta v(x)=0.
%\label{eq:asy-A}
%\end{align}
%Fortunately the precise solution for \eqref{eq:asy-A} is obtained in Proposition \ref{propch1} in Section 4.
%Hence, we get the asymptotic solution for \eqref{ORID:HJB2} as $x\rightarrow \infty$ as follows
%\begin{align*}
%v(x)\sim M\frac{x^p}{p},
%\end{align*}
From \eqref{transform}, \eqref{h} and \eqref{ORID-value-function-noncheap}, if $A$ is strictly positive, we get
\begin{align*}
v(x)\sim K_0 x^{\frac{\alpha-1}{\alpha}},
\end{align*}
where $K_0$ is some fixed constant and $a(x) \sim b(x)$ means that $\lim\limits_{x\rightarrow \infty}a(x)/b(x)=1$.
However, the following upper bound
$$V(x)\leq \mathbb{E}\left[\int_0^{\infty} e^{-\beta s}u(x+(\theta+1)as+bB_s^+)\;d s\right]$$
shows that there exists constant $K>0$ such that
\begin{equation}\label{upperunform}
V(x)\leq Kx^p.
\end{equation}
Therefore, if we want to have $v(x)=V(x)$, by Assumption \ref{thm1a} we must choose $A$ to be equal zero.
Thus we obtain the representation of the function $v$ given in \eqref{ORID-value-function-noncheapa}.
\par
In the following, we will derive the value of $\xi_0$ which is given in \eqref{ORID-xi-0} and satisfies $g(\xi_0)=0$. Indeed, if we substitute $\xi_0$ into (\ref{ORID-value-function-noncheap}) recalling that $v(0)=0$, then:
\begin{align*}
-(\eta-\theta)a+\frac{\eta^2a^2}{2b^2}g'(\xi_0)+\frac{1-p}{p}e^{\frac{\xi_0}{1-p}}=0.
\end{align*}
Since $g'(\xi_0)=Be^{\frac{\xi_0}{1-p}}+D$ we obtain expression \eqref{ORID-xi-0}.
Note that $h(\xi)=g'(\xi)$ and $A=0$ we get
\eqref{ORID-fun:g} from \eqref{h}.

We will prove that $0<q(x)\leq 1$ as it should be from the construction of the reinsurance policy.
Substituting $v'(x)=e^{-\xi}$ and $v''(x)=-e^{-\xi}/{g'(\xi)}$ into (\ref{ORID-rein-prop-rate}), we get
\begin{align*}
q(x)=1-\frac{\eta a}{b^2}g'(\xi)\leq 1.
\end{align*}
To satisfy required condition
$q(x)>0$ we need to have the following inequality
\begin{align*}
\frac{\eta a}{b^2}g'(\xi)<1
\end{align*}
satisfied.
We will show that $\xi^*$ given in \eqref{xistar} is the unique solution of the equation:
\begin{align*}
\frac{\eta a}{b^2}g'(\xi)=1.
\end{align*}
It is easily verified that $\xi^*$ solves the above equation.
Moreover, since
\begin{align*}
g''(\xi)=\frac{B}{1-p}e^{\frac{\xi}{1-p}}>0,
\end{align*}
the function $g'(\xi)$ is strictly increasing. Hence, $\xi^*$ solves above equation uniquely if
\begin{align}\label{dodwar}
\frac{\eta a}{b^2}g'(\xi_0)<1.
\end{align}
Plugging (\ref{ORID-xi-0}) into the left side of inequality \eqref{dodwar} we can conclude that
\begin{align*}
\frac{\eta a}{b^2}g'(\xi_0)&=\frac{2}{\eta a}\left[-\frac{1-p}{1-\alpha+\alpha p}e^{\frac{\xi_0}{1-p}}+\frac{(\eta-\theta)a}{\alpha}\right] \nonumber\\
&=\frac{2}{\eta a}\left[-\frac{1-p}{1-\alpha+
\alpha p}\frac{(\eta-\theta)(\alpha-1-\alpha p)ap}{\alpha (1-p)^2}
+\frac{(\eta-\theta)a}{\alpha}\right]\nonumber\\
&=\frac{2(\eta-\theta)}{\alpha\eta(1-p)}. \nonumber
\end{align*}
Therefore, if $2(\eta-\theta)/(\alpha\eta(1-p))<1$, that is $(2-\alpha(1-p))\eta-2\theta<0$, then
\eqref{dodwar} holds true.

Finally, we will check that $v$ is indeed a $C^2$ function solving the HJB equation \eqref{ORID-HJB} on $(0, x^*].$
From \eqref{ORID-value-function-noncheap} we have
\begin{align*}
v^{\prime}(x)
&=v^{\prime}(\xi)\cdot \frac{1}{g^{\prime}(\xi)}\\
&=\left\{-\frac{e^{-\xi}}{\beta}\left[-(\eta-\theta)a
+\frac{\eta^2a^2}{2b^2}g^{\prime}(\xi)+\frac{1-p}{p}e^{\frac{\xi}{1-p}}\right]
+\frac{e^{-\xi}}{\beta}\left[\frac{\eta^2a^2}{2b^2}g^{\prime\prime}(\xi)+\frac{1}{p}e^{\frac{\xi}{1-p}}\right]
\right\}\cdot \frac{1}{g^{\prime}(\xi)}\\
&=\frac{e^{-\xi}}{\beta g^{\prime}(\xi)}\left\{
(\eta-\theta)a-\frac{\eta^2a^2}{2b^2}g^{\prime}(\xi)+e^{\frac{\xi}{1-p}}
+\frac{\eta^2a^2}{2b^2}g^{\prime\prime}(\xi)\right\}\\
&=e^{-\xi},
\end{align*}
where the last equality is obtained by \eqref{eq:g-xi-1},
and
\begin{align*}
v^{\prime\prime}(x)=-\frac{e^{-\xi}}{g^{\prime}(\xi)}<0.
\end{align*}
Therefore, $v$ is increasing and strictly concave.
Substituting the values of $v^{\prime}$ and $v^{\prime\prime}$ into the left hand side of HJB equation \eqref{ORID-HJB}, we have
\begin{align}
&\max_{0\leq q\leq1,c\geq 0}\bigg\{(\theta-\eta q)a v{'}(x)
 +\frac{1}{2}b^2(1-q)^2v{''}(x)-c v{'}(x)+u(c)-\beta v(x)
 \bigg\} \nonumber\\
&=\max_{0\leq q\leq1,c\geq 0}\bigg\{(\theta-\eta q)a e^{-\xi}
 -\frac{1}{2}b^2(1-q)^2\frac{e^{-\xi}}{g^{\prime}(\xi)}-c e^{-\xi}+u(c)-\beta v(x)
 \bigg\}. \label{HJBlno}
\end{align}
Since $u(\cdot)$ is an increasing and concave function and $0<g^{\prime}(\xi)\leq {b^2}/{\eta a}$ for $\xi_0<\xi\leq\xi^*$,
therefore the maximizers $q^*$ and $c^*$ are given by
\begin{align*}
q^*(x)&=1-\frac{\eta a}{b^2}g^{\prime}(\xi), \\
c^*(x)&=e^{\frac{\xi}{1-p}}
\end{align*}
with $x=g(\xi)$.
Hence, the equation \eqref{HJBlno} can be rewritten as follows:
\begin{align}
-(\eta-\theta)a e^{-\xi}+\frac{\eta^2a^2}{2b^2}e^{-\xi}g^{\prime}(\xi)+\frac{1-p}{p}e^{\frac{p\xi}{1-p}}-\beta v(x)=0. \label{HJBlno1}
\end{align}
The verification of \eqref{HJBlno1} is straightforward involving \eqref{ORID-value-function-noncheapa} and \eqref{ORID-fun:g}.
\end{proof}
\begin{remark}\rm
To solve the equation \eqref{ORID:HJB2} we can also use the Legendre transform method. That is, we define
\begin{align}
\hat{v}(z)=\max_{x>0}\{v(x)-xz\} \label{le-tra-hv}
\end{align}
where the maximizing value of $x$ in \eqref{le-tra-hv} is
the inverse function $I(z)$ of $v'$.
Then $v(x)$ can be recovered by
 \begin{align}
v(x)=\min_{z>0}\{\hat{v}(z)+xz\}  \label{le-tra-v}
\end{align}
and the minimizing value of $z$ in \eqref{le-tra-v} equals $v'(x)$.
Substituting $x=I(z)$ into \eqref{ORID:HJB2} produces:
\begin{align}
\frac{\eta^2a^2}{2b^2}z^2\hat{v}{''}(z)+\beta z\hat{v}{'}(z)-\beta\hat{v}(z)+\frac{1-p}{p}z^{-\frac{p}{1-p}}
-(\eta-\theta)az=0. \label{eq:le-tra}
 \end{align}
 Taking derivatives on both sides of \eqref{eq:le-tra} gives:
 \begin{align}
 \frac{\eta^2a^2}{2b^2}z^2\hat{v}{'''}(z)+\left(\frac{\eta^2a^2}{b^2}+\beta\right)z\hat{v}{''}(z)-
 z^{-\frac{1}{1-p}}-(\eta-\theta)a=0. \label{le-tr-de}
 \end{align}
 Denote $h_1(z)=\hat{v}{'}(z), z=e^{-t}$. Then \eqref{le-tr-de} can be rewritten as
 \begin{align}
 \frac{\eta^2a^2}{2b^2}h^{\prime\prime}_1(t)-\left(\frac{\eta^2a^2}{2b^2}+\beta\right)h^{\prime}_1(t)
 -e^{\frac{t}{1-p}}-(\eta-\theta)a=0.
 \label{le-tr-1}
 \end{align}
 By substituting $h_2(t)=h^{\prime}_1(t)$ into \eqref{le-tr-1} we get
 \begin{align}
 h^{\prime}_2(t)-\alpha h_2(t)-\frac{2b^2(\eta-\theta)}{\eta^2a}-\frac{2b^2}{\eta^2a^2}e^{\frac{t}{1-p}}=0.
 \label{le-tr-2}
 \end{align}
 The solution of \eqref{le-tr-2} is then given by:
 \begin{align}
 h_2(t)=-Ae^{\alpha t}-Be^{\frac{t}{1-p}}-D.
 \end{align}
Moreover, from \eqref{ORID:HJB2} and the definition of the Legendre transform, we can find that
\begin{align*}
\beta v(x)&=-(\eta-\theta)az+\frac{\eta^2a^2}{2b^2}z^2\hat{v}{''}(z)+\frac{1-p}{p}z^{-\frac{p}{1-p}}
\\
&=-(\eta-\theta)ae^{-t}-\frac{\eta^2a^2}{2b^2}e^{-t}h_2(t)+\frac{1-p}{p}e^{\frac{pt}{1-p}},
\end{align*}
which coincides with \eqref{ORID-value-function-noncheap}.
%Therefore, we provide another approach to find the solution of \eqref{ORID:HJB2}.
\end{remark}

We will now consider the value function $v$ on the interval $(x^*,\infty)$ on which $q^*(x)\equiv 0$.
It is easy to verify that under the assumption that $q(x)=0$ on $(x^*,\infty)$ the HJB equation is equivalent to the following equation:
\begin{align}
\theta a v'(x)+\frac{1}{2}b^2v''(x)+\frac{1-p}{p}(v'(x))^{-\frac{p}{1-p}}-\beta v(x)=0 \label{ORID-eq:v2}
\end{align}
with the boundary conditions
\begin{align}
v'(x^*)&=e^{-\xi^*},  \label{ORID-v2-boundary-1}
\\
\frac{v'(x^*)}{v''(x^*)}&=-\frac{b^2}{\eta a}. \label{ORID-v2-boundary-2}
\end{align}
It remains now to verify that if $v$ solves the equation (\ref{ORID-eq:v2}), then the optimal reinsurance proportion indeed equals zero.
\begin{prop}\label{ORID-prop:no-rei}
Assume that Assumption \ref{thm1a} holds. Suppose also that $v$ on $(x^*,\infty)$ solves (\ref{ORID-eq:v2})-(\ref{ORID-v2-boundary-2}). Then
$v'(x)>0$ and  $v''(x)<0$ on $(x^*,\infty)$ and the maximum in
\begin{align}
\max_{0\leq q\leq 1}\left\{(\theta-\eta q)a v'(x)+\frac{1}{2}b^2(1-q)^2v''(x)\right\}=0 \label{ORID-eq:rein}
\end{align}
is attained for $q=0$ in this case.
\end{prop}
\begin{proof}
We first show that $v'(x)>0$. Assume a contrario that there exists a point
$\breve{x}$ on $(x^*,\infty)$ such that $v'(\breve{x})<0$. This implies that we can find a point $\ddot{x}, x^*<\ddot{x}<\breve{x}$
satisfying $v'(\ddot{x})=0$. Then
\begin{align*}
v''(\ddot{x}-)=\lim_{x \uparrow\ddot{x}}{\frac{v'(x)-v'(\ddot{x})}{x-\ddot{x}}}<0.
\end{align*}
Substituting $\ddot{x}$ into equation (\ref{ORID-eq:v2}) gives:
\begin{align*}
\frac{1}{2}b^2v''(\ddot{x}-)=\beta v(\ddot{x})>0
\end{align*}
which is a contradiction. Hence $v$ is indeed an increasing function.
\par Function $v$ is also concave, that is, $v''(x)<0$.
Indeed, denote
\begin{align*}
y(v)=v'(x)>0.
\end{align*}
Then $v''(x)=y(v)y'(v)$ and
plugging it into (\ref{ORID-eq:v2}) produces
\begin{align}
\theta ay(v) +\frac{1}{2}b^2y(v)y'(v)+\frac{1-p}{p}y(v)^{-\frac{p}{1-p}}-\beta v=0 \label{ORID-eq:y}
\end{align}
with the boundary conditions $y'(v^*)=-\eta a/b^2<0$ and $y(v^*)=e^{-\xi^*}$.
Moreover, the equation (\ref{ORID-eq:y}) can be rewritten as follows:
\begin{align}
y'(v)=\frac{2}{b^2}\left(\beta \frac{v}{y(v)}-\theta a-\frac{1-p}{p}y(v)^{-\frac{1}{1-p}}\right). \label{ORID-eq:y-ch}
\end{align}
According to Theorem 3 of \citet{hubalek2004optimizing}, the differential equation (\ref{ORID-eq:y-ch}) has precisely one decreasing convex solution.
Therefore,
\begin{equation}\label{once}
y'(v)<0 \quad\mbox{and}\quad y''(v)>0.
\end{equation}
Hence, $v''(x)<0$.
\par Now we will show that (\ref{ORID-eq:rein}) attaines maximum at zero, that is
\begin{align}\label{counterpart}
-\frac{\eta a}{b^2} \frac{v'(x)}{v''(x)}>1
\end{align}
or that
\begin{align}
y'(v)>-\frac{\eta a}{b^2}. \label{ORID-ineq:y}
\end{align}
\par Denote
\begin{align*}
q(v)=\left(\theta-\frac{\eta}{2}\right)ay(v)+\frac{1-p}{p}y(v)^{-\frac{p}{1-p}}.
\end{align*}
Since $y(v)$ satisfies equation (\ref{ORID-eq:y}), thus $q(v)$ can be rewritten as:
\begin{align*}
q(v)=\beta v-\frac{\eta a}{2} y(v)-\frac{1}{2}b^2y(v)y'(v).
\end{align*}
Therefore demonstrating (\ref{ORID-ineq:y}) is equivalent to showing that
\begin{align}
q(v)<\beta v\quad\mbox{for $v>v^*$.}\label{demonstr}
\end{align}
Due to the boundary condition $y'(v^*)=-\eta a/b^2$, we can observe that $q(v^*)=\beta v^*$.
We also claim that \begin{equation}\label{qstar}q'(v^*)<\beta.\end{equation} Indeed,
\begin{align*}
q'(v^*)&=\left(\theta-\frac{\eta}{2}\right)ay'(v^*)-y(v^*)^{-\frac{1}{1-p}}y'(v^*) \\
&=-\frac{\eta a^2}{b^2}\left(\theta-\frac{\eta}{2}\right)+\frac{\eta a}{b^2}e^{\frac{\xi^*}{1-p}}\\
&=-\frac{\eta a^2}{b^2}\left(\theta-\frac{\eta}{2}\right)+\frac{\eta a}{b^2}\frac{1}{B}\left(\frac{b^2}{\eta a}-D\right)\\
&=\beta-\frac{\eta a^2}{2 \alpha b^2(1-p)}(\eta \alpha-2(\eta-\theta)),
\end{align*}
where the second equality is obtained by substituting the values of $y(v^*)$ and $y'(v^*)$. The third equality follows from the form of $\xi^*$ and
the last equality is obtained by plugging the values of $B, D$ and $\alpha$.
Now the inequality $\eta \alpha-2(\eta-\theta)>\alpha p\eta>0$ implies \eqref{qstar}.
\par To prove \eqref{demonstr}, suppose a contrario that there exists $\tilde{v}\in (v^*, \infty)$ such that $q(\tilde{v})=\beta \tilde{v}$, that is $y'(\tilde{v})=-\eta a/b^2$,
and that $q(v)<\beta v$ on $(v^*, \tilde{v})$.
Since $q(\tilde{v})=\beta \tilde{v}$ and $q(v)<\beta v$ on $(v^*, \tilde{v})$, we have $q'(\tilde{v})>\beta$.
This implies that
\begin{align}
\left(\theta-\frac{\eta}{2}\right)ay'(\tilde{v})-y(\tilde{v})^{-\frac{1}{1-p}}y'(\tilde{v})>\beta. \label{ORID-eq:com-q}
\end{align}
On the other hand, by taking derivatives with $v$ on both sides of (\ref{ORID-eq:y}) and by substituting the value of $\tilde{v}$, we obtain
\begin{align*}
&\theta a y'(\tilde{v})+\frac{1}{2}b^2(y'(\tilde{v}))^2+\frac{1}{2}b^2y(\tilde{v})y''(\tilde{v})-
y(\tilde{v})^{-\frac{1}{1-p}}y'(\tilde{v})-\beta=0 \\
\end{align*}
and therefore
\begin{align*}
&\theta a y'(\tilde{v})-
y(\tilde{v})^{-\frac{1}{1-p}}y'(\tilde{v})=\beta-\frac{1}{2}b^2(y'(\tilde{v}))^2-\frac{1}{2}b^2y(\tilde{v})y''(\tilde{v}).
\end{align*}
Now, substituting above identity into (\ref{ORID-eq:com-q}) gives
\begin{align*}
y''(\tilde{v})<0
\end{align*}
because $y'(\tilde{v})=-\eta a/b^2$.
This is a contradiction with the fact that by \eqref{once} $y(v)$ is a decreasing convex function.
This completes the proof.
\end{proof}
The equation (\ref{ORID-eq:v2}) satisfied by $v$ for $x>x^*$ is similar to equation (15) in  \citet{hubalek2004optimizing}.
Therefore, we can identify the asymptotic solution of the value function and the corresponding optimal strategies
as it was also done in \citet{hubalek2004optimizing}.
%We will write $\chi(x)\sim \varpi(x)$ if $\lim_{x\to\infty} \chi(x)/\varpi(x)=1$.
\begin{theorem}\label{th:asyno}
The asymptotic behaviors of $v(x), c(x), q(x)$, as $x\rightarrow \infty$ , are given by:
\begin{align*}
&v(x)\sim \left(\frac{1-p}{\beta}\right)^{1-p}\frac{x^p}{p}, \\
&c(x)\sim  \frac{\beta}{1-p}x, \\
&q(x)\equiv 0.
\end{align*}
\end{theorem}
Summarizing Propositions \ref{nocheap1}, \ref{ORID-prop:no-rei} and Theorem \ref{th:asyno}
gives the first main result of this paper.
\begin{theorem}\label{thm1}
Assume Assumption \ref{thm1a}. Then the value function $V(x)$ is given by
\begin{equation}\label{rhs1}
V(x)=
\begin{cases}
\frac{e^{-\xi}}{\beta}\left[-(\eta-\theta)a+\frac{\eta^2a^2}{2b^2}Be^{\frac{\xi}{1-p}}+
\frac{\eta^2a^2}{2b^2}D+\frac{1-p}{p}e^{\frac{\xi}{1-p}}\right]
,&~~~~0<x\leq {x^*},
\\v(x)
,&~~~~x>x^*,
\end{cases}
\end{equation}
with $B$ and $D$ given in (\ref{ORID-B}) and (\ref{ORID-D}) and
$x=g(\xi)$, where $g$ is defined in (\ref{ORID-fun:g}). Moreover,
$x^*$ is given in (\ref{ORID-x^*})
and $v(x)$ solves (\ref{ORID-eq:v2})-(\ref{ORID-v2-boundary-2}).
The corresponding optimal dividend strategy is
\begin{eqnarray}
c(x)=(V'(x))^{-\frac{1}{1-p}}\nonumber
\end{eqnarray}
and the optimal reinsurance proportion is
\begin{equation}
q(x)=
\begin{cases}
1-\frac{\eta a}{b^2}g'(\xi)
,&~~~~0<x\leq {x^*},
\\0
,&~~~~x>x^*.
\end{cases}\nonumber
\end{equation}
\end{theorem}
\begin{proof}
Note that $c(x)$ and $q(x)$ have bounded continuous derivatives on
$(0, x^*]$.  For $x>x^*$, we have $q(x)\equiv 0$ and $c(x)$ is asymptotically equivalent to $\beta x/(1-p)$ as $x$ tends to $\infty$.
Therefore, their derivatives are bounded on $(x^*, \infty)$. Hence, under the optimal strategy, the classical result in SDE guarantees the existence
and uniqueness of the solution for SDE \eqref{control}.
Moreover, from Propositions \ref{nocheap1} and \ref{ORID-prop:no-rei} it follows that the right hand side $v(x)$
of \eqref{rhs1} is a nonnegative and concave solution of HJB equation
\eqref{ORID-HJB} with the boundary condition \eqref{ORID-HJB-boundary}. Similarly to the proof of Lemma 1,
applying It${\rm \hat{o}}$'s formula to the process $w(t,X^{\pi^*}_t)$ with the
corresponding strategies identified in Propositions \ref{nocheap1} and \ref{ORID-prop:no-rei}:
\[\pi^*=(q^*_t, c^*_t)=\left(q(X^{\pi^*}_t), c(X^{\pi^*}_t)\right)\]
produces:
\begin{align}
&w\left(t\wedge \tau_{\xi_1}\wedge \tau_{\xi_2},X^{\pi^*}_{t\wedge \tau_{\xi_1}\wedge \tau_{\xi_2}}\right)
\nonumber\\
&=v(x)+\int^{t\wedge \tau_{\xi_1}\wedge \tau_{\xi_2}}_0 e^{-\beta s}
\left\{(\theta-\eta q^*_s)a v{'}(X^{\pi^*}_s)
 +\frac{1}{2}b^2(1-q^*_s)^2v{''}(X^{\pi^*}_s)-c^*_s v{'}(X^{\pi^*}_s)-\beta v(X^{\pi^*}_s)\right\}ds \nonumber\\
&~~~~+\int^{t\wedge \tau_{\xi_1}\wedge \tau_{\xi_2}}_0 e^{-\beta s}v{'}(X^{\pi^*}_s)b(1-q^*_s)dB_s\nonumber\\
&=v(x)-\int^{t\wedge \tau_{\xi_1}\wedge \tau_{\xi_2}}_0 e^{-\beta s}
u(c^*_s)ds+\int^{t\wedge \tau_{\xi_1}\wedge \tau_{\xi_2}}_0 e^{-\beta s}v{'}(X^{\pi^*}_s)b(1-q^*_s)dB_s,
\label{verino}
\end{align}
where the last equality is obtained by the HJB equation \eqref{ORID-HJB}.
Taking expectation on both sides of \eqref{verino}, we get
\begin{align*}
\mathbb{E}_x\left[w\left(t\wedge \tau_{\xi_1}\wedge \tau_{\xi_2},X^{\pi^*}_{t\wedge \tau_{\xi_1}\wedge \tau_{\xi_2}}\right)+\int^{t\wedge \tau_{\xi_1}\wedge \tau_{\xi_2}}_0 e^{-\beta s}
u(c^*_s)ds\right]= v(x).
\end{align*}
We claim that
$\left\{w(t, X^{\pi^*}_{t})\right\}_{t\geq 0}$ is uniformly integrable.
Then letting $\xi_1 \downarrow 0$, $\xi_2\rightarrow \infty$ and $t\rightarrow \infty$,
and by applying the dominated convergence theorem and the monotone convergence theorem, we have
\[\mathbb{E}_x\left[\int^{\tau}_0 e^{-\beta s}
u(c^*_s)ds\right]=v(x).\]
Thus $v(x)\leq V(x)$ from the definition of the value function.
Combining above inequality with the result obtained in Verification Lemma \ref{ORID-verif-lemma} completes the proof.
In order to verify the uniformly integrability of $w(t, X^*_t)$, it is sufficient to show that the maximum function
$w^*(T, X^{\pi^*}_{T})=\sup_{0\leq t\leq T}w(t, X^{\pi^*}_{t})$ is integrable for each $T\geq 0$.
By the concavity of $v$, we have
\begin{equation}
v(x)\leq v^\prime(0)x.\nonumber
\end{equation}
Thus,
\begin{align}
\mathbb{E}_x w^*(T, X^{\pi^*}_{T})\leq \mathbb{E}_x\left[\sup\limits_{0\leq t\leq T} v\left(X^{\pi^*}_t\right)\right]
\leq v^\prime(0)\mathbb{E}_x\left[\sup\limits_{0\leq t\leq T} X^{\pi^*}_t\right]. \label{uniforno}
\end{align}
We introduce now a new stochastic process $Y_t$ having the following dynamics:
\begin{equation}
dY_t=\theta a dt+b(1-q^*_t)d B_t \label{sdey}
\end{equation}
with the same initial value as that of the reserve process $X^{\pi^*}_t$, that is, $Y_0=x$. Then one can observe that $\{Y_t\}_{0\leq t\leq T}$ is a submartingale and $X^{\pi^*}_t\leq Y_t$ for any $t \geq 0$. We denote $Y^*_T=\sup_{0\leq t\leq T}Y_t$.
Then,
\begin{align}
\mathbb{E}_x\left[\sup\limits_{0\leq t\leq T} X^{\pi^*}_t\right]
\leq \mathbb{E}_x\left(Y^*_T\right)
\leq\sqrt{\mathbb{E}_x \left(Y^*_T\right)^2}
\leq 2\sqrt{\mathbb{E}_x \left(Y_T\right)^2}
\leq K_1T+K_2 \sqrt{T}<\infty, \label{ieqy}
\end{align}
where $K_1, K_2$ are some fixed constants, the second inequality is obtained by Cauchy-Swartz inequality and the third inequality is derived by Doob's maximal inequality for submartingale. Therefore, in view of \eqref{uniforno} and \eqref{ieqy}, we can conclude that $\left\{w(t, X^{\pi^*}_{t})\right\}_{t\geq 0}$ is uniformly integrable.
This completes our proof.
\end{proof}

{\section{Cheap reinsurance}\label{sec:cheap}}
\indent
In this section we consider the case $\eta=\theta$, that is that cheap reinsurance holds true.
The corresponding HJB equation could be rewritten as follows:
\begin{eqnarray}
 & &\max_{0\leq q\leq1,c\geq 0}\bigg\{\theta(1-q)a v'(x)
 +\frac{1}{2}b^2(1-q)^2v{''}(x)-c v{'}(x)+u(c)-\beta v(x)
 \bigg\}=0.
 \label{ORID-ch-rein-HJB}
 \end{eqnarray}
\indent As argued in the previous section, we suppose
that HJB equation (\ref{ORID-ch-rein-HJB}) has an increasing concave solution $v$. Then, for $q$ without
restriction, the left hand side of (\ref{ORID-ch-rein-HJB}) attains its maximum at
\begin{eqnarray}
q(x)&=&1+\frac{\theta a}{b^2}\frac{v{'}(x)}{v{''}(x)}\label{ORID-rein-prop-rate2}
\end{eqnarray}
and \[c(x)=(v{'}(x))^{-\frac{1}{1-p}}. \]
Similarly, as in the non-cheap reinsurance case, we will find the point $x^*$ such that
$q(x)=0$, for all $x>x^*$. Therefore, we have to consider here also two intervals $(0,x^*]$ and $(x^*, \infty)$.
At the beginning we will analyze the case of $(0,x^*]$ on which $q(x)>0$.
We recall the definition of $\alpha$ in \eqref{ORID-alpha} with $\eta=\theta$.
Let
\begin{align}
x^*=\frac{b^2(1-p)}{\theta a}. \label{ORID-ch-rei-x^*}
\end{align}
\begin{prop}\label{propch1}
Assume that Assumption \ref{thm2a} holds.
Then for $x\in(0,x^*]$,
\begin{align*}
v(x)=M\frac{x^p}{p},
\end{align*}
where
\begin{align}
M=\left[\frac{\beta}{1-p}-\frac{\theta^2 a^2}{2b^2}\frac{p}{(1-p)^2}\right]^{-(1-p)}. \label{ORID-ch-rei-M}
\end{align}
\end{prop}
\begin{proof}
Note that (\ref{ORID-ch-rein-HJB}) can be rewritten as follows:
\begin{equation}
-\frac{\theta^2a^2}{2b^2}\frac{(v'(x))^2}{v''(x)}+\frac{1-p}{p}(v'(x))^{-\frac{p}{1-p}}-\beta v(x)=0.
\label{ORID:ch-rein-HJB2}
\end{equation}
The general solution of ODE (\ref{ORID:ch-rein-HJB2}) is given by (see \citet{merton1969lifetime} for details):
\begin{eqnarray}
v(x)=\frac{M(x-N)^p}{p}+Q. \nonumber
\end{eqnarray}
Then
\begin{align*}
v'(x)=M(x-N)^{p-1}, ~~~~~~~~~~~v''(x)=-M(1-p)(x-N)^{p-2}.
\end{align*}
Plugging it back into (\ref{ORID:ch-rein-HJB2}) produces
$Q=0$ and $M$ given in \eqref{ORID-ch-rei-M}.
Note that $\alpha>{1}/{(1-p)}$,
that is, $1+{2b^2\beta}/{\theta^2 a^2}>1/{(1-p)},$
which implies that $\beta>{\theta^2 a^2 p}/{(2b^2(1-p))}.$
Hence, $M>0$.
Finally, from the boundary condition $v(0)=0$ we get $N=0$. Thus the proof is completed.
\end{proof}

We will now consider the interval $(x^*,\infty)$ on which we will have $q(x)\equiv 0$.
Similarly as the non-cheap reinsurance case, from HJB equation (\ref{ORID-ch-rein-HJB}) we can derive the equation satisfied by the value function on this interval:
\begin{align}
\theta a v'(x)+\frac{1}{2}b^2v''(x)+\frac{1-p}{p}(v'(x))^{-\frac{p}{1-p}}-\beta v(x)=0 \label{ORID-eq:ch-rei-v2}
\end{align}
with the boundary conditions
\begin{align}
v'(x^*)&={M}(x^*)^{p-1},   \label{ORID-ch-rei-v2-boundary-1}
\\
\frac{v'(x^*)}{v''(x^*)}&=-\frac{x^*}{1-p}. \label{ORID-ch-rei-v2-boundary-2}
\end{align}
Note that the second-order ODE (\ref{ORID-eq:ch-rei-v2}) is the same as ODE (\ref{ORID-eq:v2}) except the two boundary conditions.
In the same way as in the proof of Proposition \ref{ORID-prop:no-rei}
one can verify that $q(x)=0$ for $x>x^*$. To do it one has to
prove inequality
\[-\frac{\theta a}{b^2}\frac{v'(x)}{v''(x)}>1\]
with $v(x)=Mx^p/p$
as a counterpart of
\eqref{counterpart}.
%if $v$ is given as the solution of (\ref{ORID-eq:ch-rei-v2}).
\begin{prop}\label{propch2}
Assume Assumptions \ref{thm2a}. Suppose $v$ on $(x^*,\infty)$ solves (\ref{ORID-eq:ch-rei-v2})-(\ref{ORID-ch-rei-v2-boundary-2}). Then $v$ satisfies $v'(x)>0$ and $v''(x)<0$ and
maximum in
\begin{align*}
\max_{0\leq q\leq 1}\left\{\theta(1-q)a v'(x)+\frac{1}{2}b^2(1-q)^2v''(x)\right\}=0
\end{align*}
is attained for $q=0$ in this case.
\end{prop}

\begin{remark}
\rm{Since (\ref{ORID-eq:ch-rei-v2}) is the same as (\ref{ORID-eq:v2}), we will have the same asymptotic results for the cheap reinsurance case as that for the non-cheap reinsurance case.}
\end{remark}

Based on the above analysis, we give the second main result of this paper.
\begin{theorem}\label{thm2}
Assume Assumption \ref{thm2a}.The value function $V(x)$ is given by
\begin{equation}
V(x)=
\begin{cases}
M\frac{x^p}{p}
,&~~~~0<x\leq {x^*},
\\v(x)
,&~~~~x>x^*,
\end{cases}\nonumber
\end{equation}
with $x^*$ and
$M$ in (\ref{ORID-ch-rei-x^*}) and (\ref{ORID-ch-rei-M}), respectively,
and $v(x)$ solves (\ref{ORID-eq:ch-rei-v2})-(\ref{ORID-ch-rei-v2-boundary-2}).
The corresponding optimal dividend strategy is
\begin{eqnarray}
c(x)=(V'(x))^{-\frac{1}{1-p}} \nonumber
\end{eqnarray}
and the optimal reinsurance proportion is
\begin{equation}
q(x)=
\begin{cases}
1-\frac{\theta a}{b^2(1-p)}x
,&~~~~0<x\leq {x^*},
\\0
,&~~~~x>x^*.
\end{cases}\nonumber
\end{equation}
\end{theorem}
\begin{proof}
Similarly like in the proof of Theorem \ref{thm1} one can prove the existence
and uniqueness of the solution for SDE \eqref{control} for the chosen optimal strategies.
Furthermore, in order to confirm that the candidate solution $v(x)$ is indeed the value function, it is sufficient to verify that
$\left\{w(t, X^{\pi^*}_{t})\right\}_{0\leq t\leq T}$ is uniformly integrable,
where $\pi^*$ is the corresponding strategy given in
Propositions \ref{propch1} and \ref{propch2}. However, since $\lim\limits_{x\to 0}v^{\prime}(x)=\infty$, we cannot use the proof of Theorem \ref{thm1} directly
and some further modification is required here. By \eqref{upperunform} we can find a large constant $K$ such that for $x>0$,
\[v(x)\leq K x^p.\]
Then,
\begin{align*}
\mathbb{E}_x w^*(T, X^{\pi^*}_{T})\leq \mathbb{E}_x\left[\sup\limits_{0\leq t\leq T} v\left(X^{\pi^*}_t\right)\right]
\leq K\mathbb{E}_x\left[\sup\limits_{0\leq t\leq T} \left(X^{\pi^*}_t\right)^p\right]
\leq K \mathbb{E}_x\left[\left(Y^*_T\right)^p\right],
\end{align*}
where $Y_t$ is given in \eqref{sdey} and $Y^*_t$ is its maximal process.
Moreover, since $p\in(0,1)$, we can find two positive integral such that $p\leq m/n$.
Thus,
\begin{align*}
\mathbb{E}_x[\left(Y^*_T\right)^p]\leq\mathbb{E}_x[\left(Y^*_T\right)^{m/n}]
\leq \left\{\mathbb{E}_x \left(Y^*_T\right)^{2m}\right\}^{\frac{1}{2n}}
\leq \left(\frac{2m}{2m-1}\right)^{\frac{m}{n}}\left\{\mathbb{E}_x \left(Y_T\right)^{2m}\right\}^{\frac{1}{2n}}
\leq H_1T^{\frac{m}{n}}+H_2 T^{\frac{m}{2n}}<\infty,
\end{align*}
where $H_1, H_2$ are some fixed constants. Therefore, $\left\{w(t, X^{\pi^*}_{t})\right\}_{0\leq t\leq T}$ is uniformly integrable.
This completes the proof.
\end{proof}
\section{Numerical examples}\label{sec:numericalexamples}
In this section, we provide numerical examples to demonstrate the results obtained in Sections \ref{sec:noncheap} and \ref{sec:cheap}. We focus on the effects of the power $p$ and the claim volatility $b$ on the optimal strategies and the threshold $x^*$ for both cheap and noncheap reinsurance cases.

 In Figure \ref{fig1:subfig} we plot the optimal dividend and optimal reinsurance proportion for surplus $x$ ranging from
 $0$ to $50$. In Figure \ref{fig1:subfig:a}, we consider the noncheap reinsurance case. According to Theorem \ref{thm1},
 we find the optimal threshold $x^*$ equals to $11.2578$, which implies that the insurer will take all the claims without buying any reinsurance when the surplus is larger than $11.2578$. From this graph we can also observe that the optimal dividend rate $c^*(x)$ is nearly linear increasing with the surplus $x$, while the optimal reinsurance proportion is decreasing with $x$.  When its surplus is sufficiently small, the insurer will divert most of its risk incurred by claims to the reinsurer, which seems intuitively reasonable. By comparison, the cheap reinsurance case has been illustrated in Figure \ref{fig1:subfig:b}, which exhibits similar phenomenon as in Figure \ref{fig1:subfig:a}. In accordance with Theorem \ref{thm2}, in this case, the optimal threshold $x^*$ equals to $20.6250$, which is larger than that in Figure \ref{fig1:subfig}. This happens because the premium paid to the reinsurer in the cheap reinsurance case is less than that paid in the noncheap one. This means that insurer will be more willing to divert the claim risk to the reinsurer in cheap reinsurance case even though its reserves are relatively large.

 \par Figure \ref{fig2:subfig} displays the optimal dividend strategy and the optimal reinsurance proportion as the functions of $p\in[0,1]$ for fixed $x=5$. From this graph,
 we can observe that the optimal dividend rate increases with increasing $p$, while the optimal reinsurance proportion decreases.
 This phenomenon can be explained by noting that the parameter $1-p$ represents the risk aversion of the insurer, which indicates the insurer's attitude towards the risk. We also find that for a fixed $p$, the optimal dividend rate and the optimal reinsurance proportion in cheap reinsurance case are larger than those in noncheap one.
 The reason behind this property is that insurer usually pays less under the cheap reinsurance contract and it tends to transfer most of the risks to the reinsurer
 paying extra money for dividends.

 Figure \ref{fig3:subfig} shows the values of $x^*$ with respect of changing parameters $p$ and $b$.
 In Figure \ref{fig3:subfig:a} we set $b=0.5$ and change the values of $p$. Note that $x^*$ decreases as $p$ increases.
 Because $p$ reflects the insurer's tolerance of the risks, as $p$ becomes larger, the insurer becomes less risk averse.
 In Figure \ref{fig3:subfig:b} we fix $p=0.01$ and  observe that  $x^*$ increases as $b$ increases.
 Indeed, this is explained by the fact that the larger values of $b$ indicate that more risks will be taken by the insurer.
 In both Figures \ref{fig3:subfig:a} and \ref{fig3:subfig:b} the values of $x^*$ in cheap reinsurance case are bigger than those in noncheap reinsurance case
 which coincides with the analysis demonstrated in Figure \ref{fig1:subfig}.

\begin{figure}[!htbp]
  \centering
  \subfigure[$\eta=0.8, \theta=0.4$]{
    \label{fig1:subfig:a} %% label for first subfigure
    \includegraphics[width=3.0in]{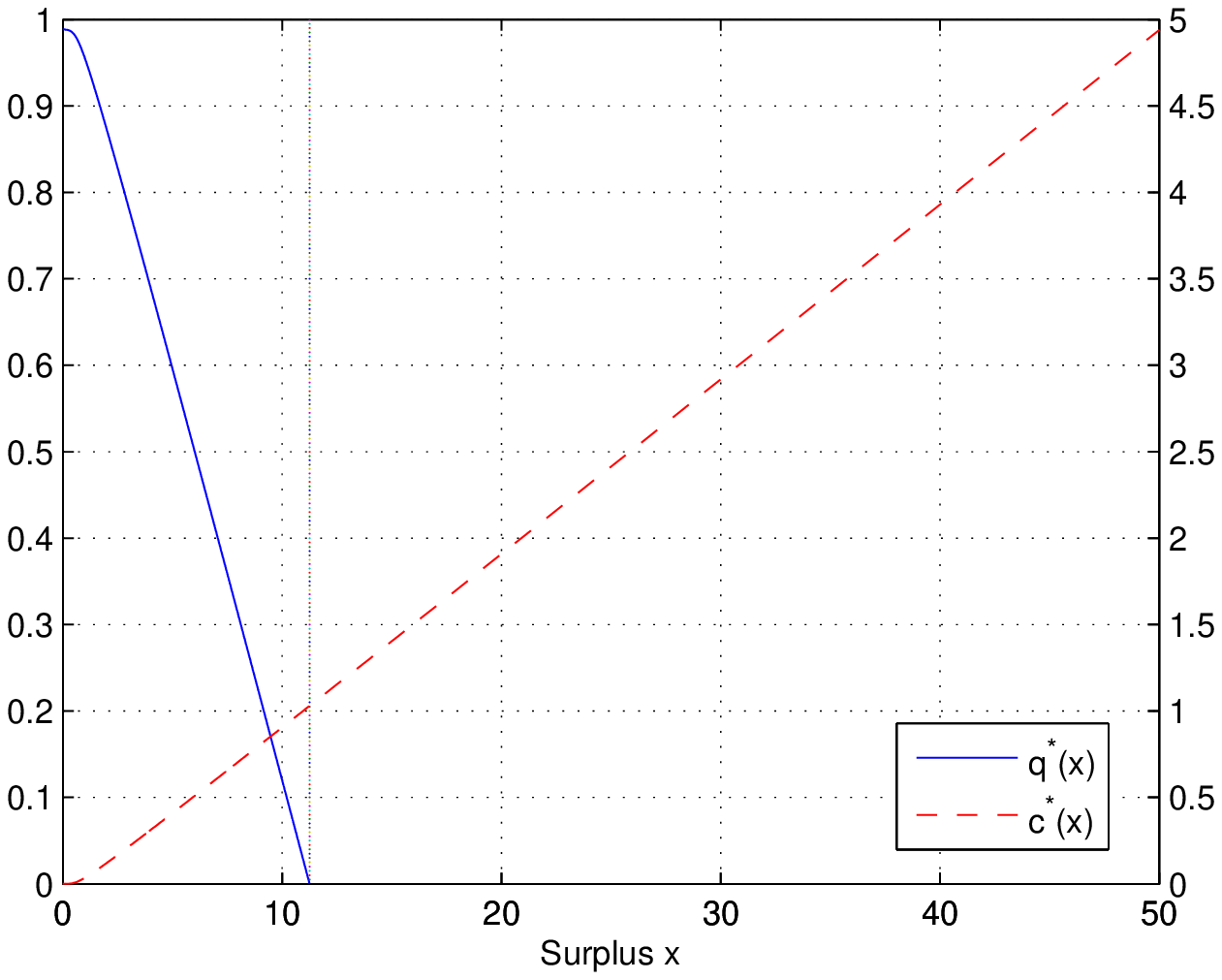}}
   \subfigure[$\eta=\theta=0.4$]{
    \label{fig1:subfig:b} %% label for second subfigure
    \includegraphics[width=3.0in]{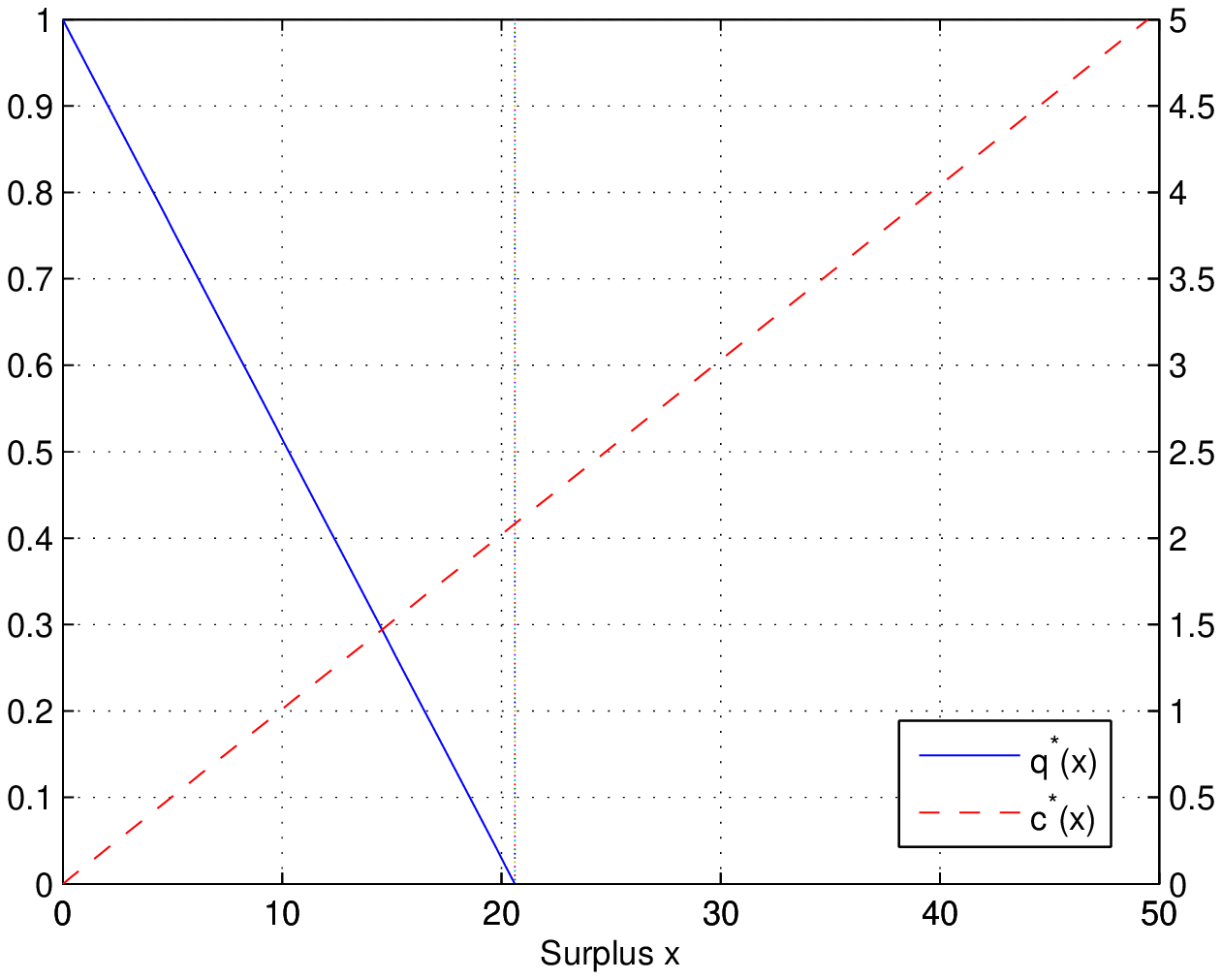}}
  \caption{$q^*(x)$ and $c^*(x)$ change with $x$ for $a=0.03, b=0.5, p=0.01, \beta=0.1.$}
  \label{fig1:subfig} %% label for entire figure
\end{figure}
\begin{figure}[!htbp]
  \centering
  \subfigure[$c^*(x)$ vs $p$]{
    \label{fig2:subfig:a} %% label for first subfigure
    \includegraphics[width=3.0in]{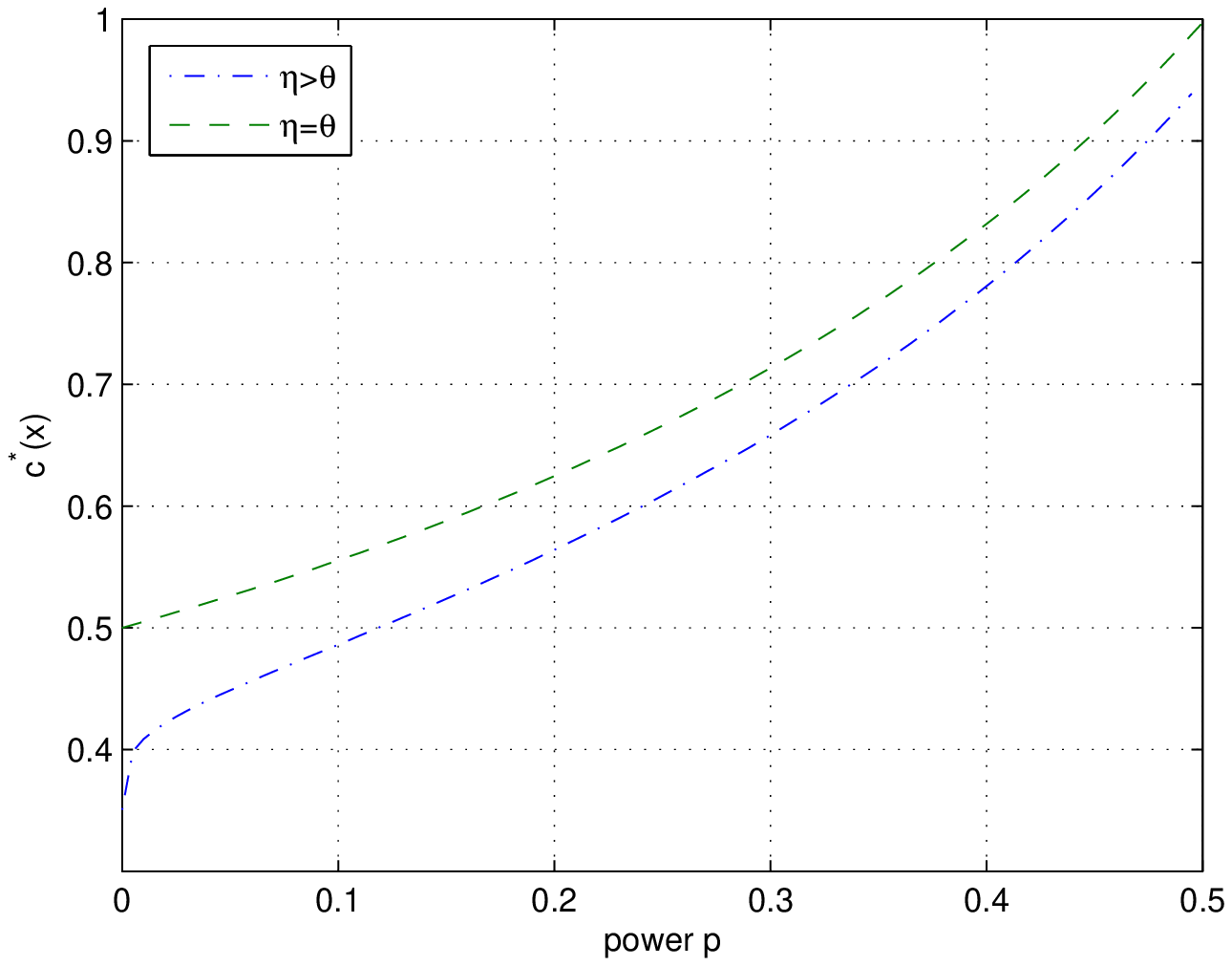}}
   \subfigure[$q^*(x)$ vs $p$ ]{
    \label{fig2:subfig:b} %% label for second subfigure
    \includegraphics[width=3.0in]{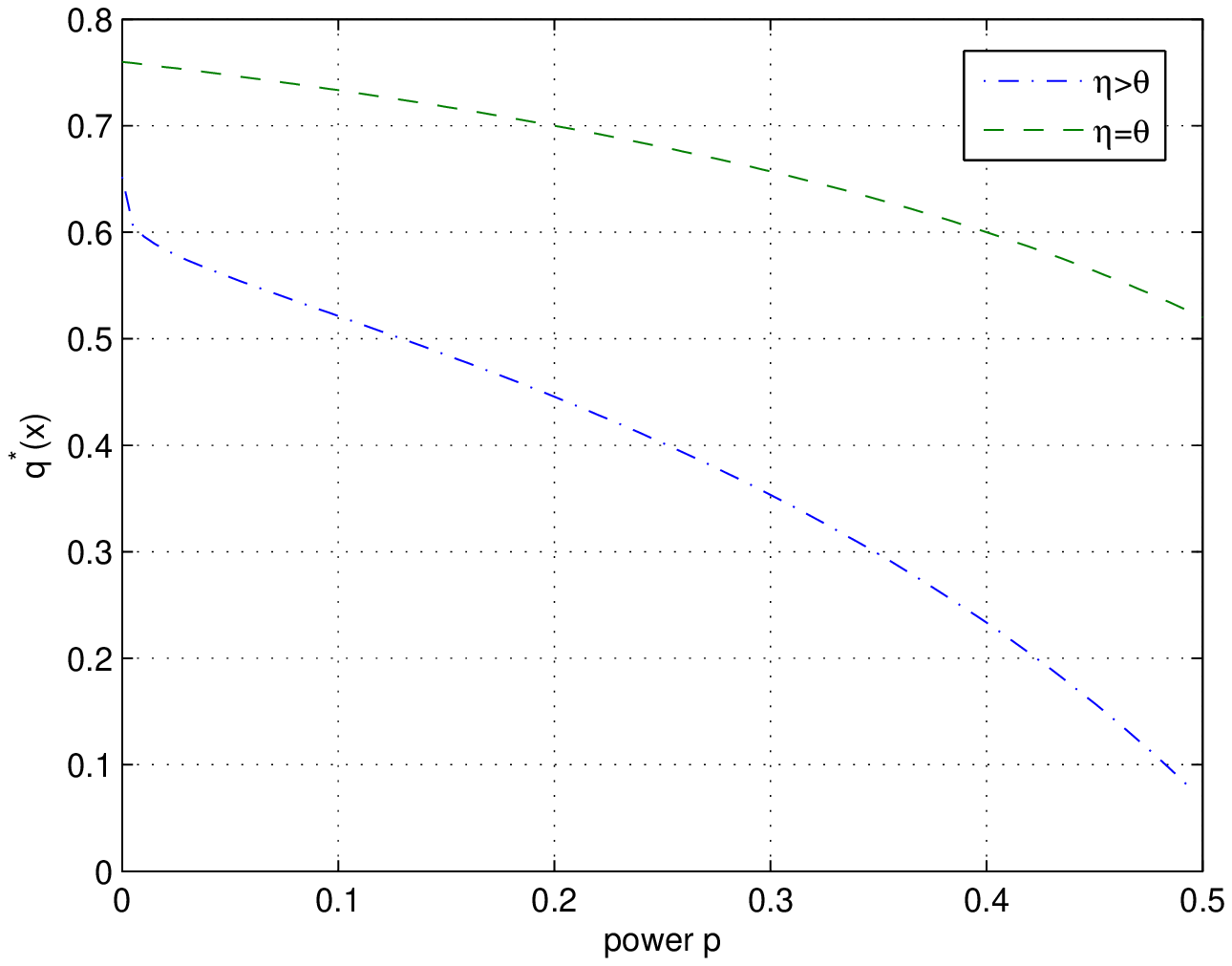}}
  \caption{$c^*(x)$ and $q^*(x)$ change with $p$ for $a=0.03, b=0.5, \beta=0.1, \eta=0.8, \theta=0.4, x=5.$}
  \label{fig2:subfig} %% label for entire figure
\end{figure}
\begin{figure}[!htbp]
  \centering
  \subfigure[$x^*$ vs $p$ for $b=0.5$]{
    \label{fig3:subfig:a} %% label for first subfigure
    \includegraphics[width=3.0in]{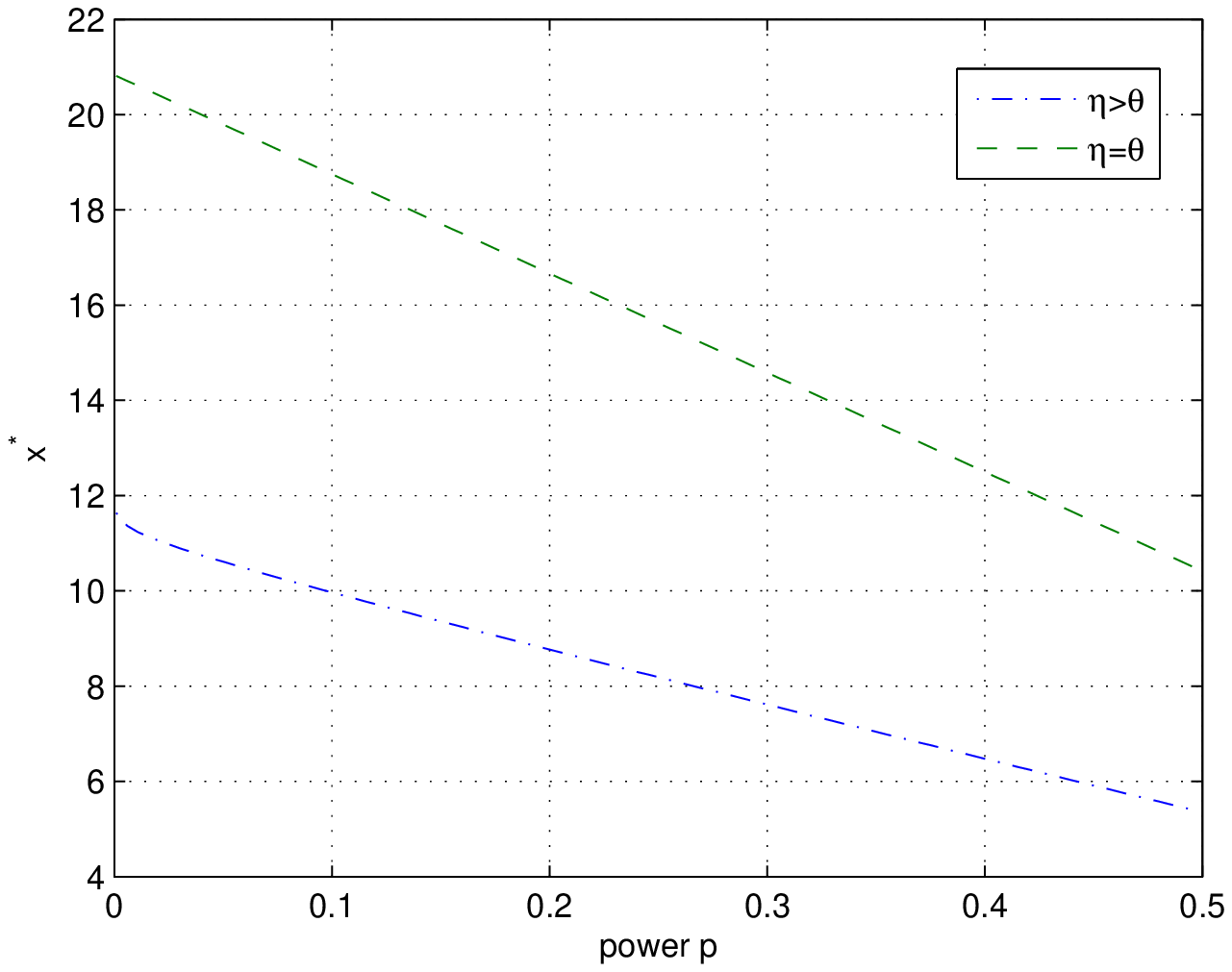}}
   \subfigure[$x^*$ vs $b$ for $p=0.01$]{
    \label{fig3:subfig:b} %% label for second subfigure
    \includegraphics[width=3.0in]{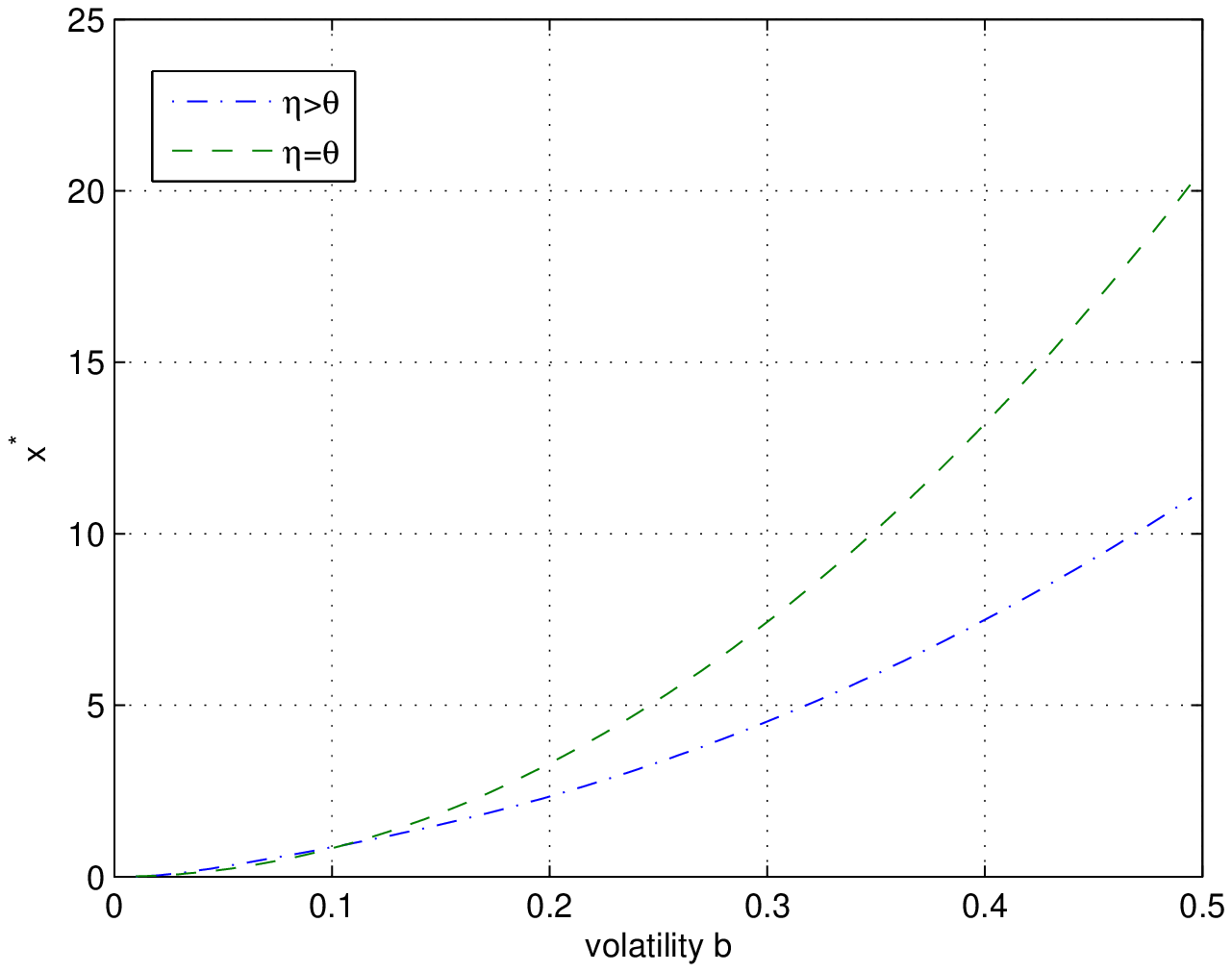}}
  \caption{$x^*$ changes with $p$ and $b$ respectively for $a=0.03, \beta=0.1, \eta=0.8, \theta=0.4.$}
  \label{fig3:subfig} %% label for entire figure
\end{figure}

\section{Conclusions}\label{sec:conclusions}
In this paper we manage to find the value function maximizing the discounted cumulative  dividends payments paid up to ruin time
where the strategy is based on choice of dividend payments and the proportion of the reinsurance policy.
We analyzed only the Constant Relative Risk Aversion utility function.
The future research will concern other utility functions.
One can also choose more general stopping time of the regulated risk process.
For example one can consider Parisian ruin time as it was done in
\citet{Czarna}.
Finally, it is also very interesting to incorporate so-called Gerber-Shiu function in the value function as it was realized in
\citet{APP2}. One can analyze the investments into risky assets as well, see \citet{PG}.
All of these problems are more complex and left for future research.

\bibliographystyle{apalike}

\end{document}